\documentclass[runningheads]{llncs}
\usepackage[T1]{fontenc}
\usepackage{graphicx}
\usepackage{amssymb}
\usepackage{amsmath}
\usepackage{todonotes}
\usepackage{comment}

\usepackage{tikz}
\usetikzlibrary{quantikz}

\newcommand{\cent}[0]{\mbox{\textcent}}
\newcommand{\dollar}{\$}
\newcommand{\mymatrix}[2]{\left( \begin{array}{#1} #2 \end{array} \right)}

\newcommand{\mypar}[1]{\left( #1 \right)}

\newcommand{\modp}{\mathtt{MOD_p}}

\usepackage{subcaption}

\begin{document}
\title{Quantum hashing algorithm implementation}
%
%
\author{Aliya Khadieva\inst{1,2}\orcidID{0000-0003-4125-2151}}
\authorrunning{A. Khadieva}
%
\institute{
Kazan Federal University, Kazan, Russia  \and University of Latvia, R\={\i}ga, Latvia 
\\
\email{aliya.khadi@gmail.com}
}
\maketitle              
\begin{abstract}
We implement a  quantum hashing algorithm which is based on a fingerprinting technique presented by Ambainis and Frievalds, 1988, on gate-based quantum computers. This algorithm is based on  a quantum finite automaton for a unary language $\mathtt{MOD_p}$, where $ \mathtt{MOD_p} = \{ a^{i \cdot p} \mid i \geq 0 \} $,  for any prime number $p$. We consider 16-qubit and 27-qubit IBMQ computers with the special graphs of qubits representing nearest neighbor architecture that is not Linear Nearest Neighbor (LNN) one. We optimize quantum circuits for the quantum hashing algorithm with respect to minimizing the number of control operators as the most expensive ones.  We apply the same approach for an optimized circuit implementation of Quantum Fourier Transform (QFT) operation on the aforementioned machines because QFT and hashing circuits have common parts.

\end{abstract}
\section{Introduction}
\label{sec:intro}

Quantum fingerprinting and quantum hashing technique are well-known computational techniques that allow to map a large data item to a much shorter string, its fingerprint that identifies the original data (with high probability).
The probabilistic technique was explored by  Freeholds \cite{Fre79}. Later Ambainis and Frievalds \cite{af98} developed quantum counterparts of the technique which allowed them to develop a quantum finite automaton with an exponentially smaller size than any classical randomized automaton for a specific language $\mathtt{MOD_p}$, where $ \mathtt{MOD_p} = \{ a^{i \cdot p} \mid i \geq 0 \} $  for any prime number $p$. 
 After that, Buhrman et al. in \cite{bcwd2001} presented an explicit definition of the quantum fingerprinting for constructing an efficient quantum communication protocol for equality testing.
 Ambainis and Nahimovs improved the quantum automaton for  $\mathtt{MOD_p}$ in \cite{an2008,an2009}.
The technique was widely used in different areas like stream processing algorithms \cite{l2009}, online algorithms \cite{kk2019disj}, automata \cite{af98}, branching programs \cite{av2013,kkk2022}, cryptography \cite{aavz2016} and others.

Here we discuss circuit implementation of quantum hashing for real quantum devices. For this reason, we focus on implementation of the automaton for $\mathtt{MOD_p}$ language because it contains the main idea of the quantum fingerprinting technique but simple enough to understand it. 



However, current quantum machines are the noisy intermediate-scale quantum machines ~\cite{preskill2018quantum}. One of the main issues for real quantum computers is computational noise caused by quantum operators and measurements. So, quantum algorithms with big depth often output incorrect results due to errors caused by noise. Therefore, the minimization of a depth of circuits is one of the main goals for researchers. Especially, it is important to minimize the number of two-qubit gates that are the most expensive gates for implementation. The CNOT gate is a basic two-qubit gate which is used in many quantum devices as a two-qubit gate. So, we focus on the number of 
CNOT gates when estimate efficiency of the circuit.

In this work, we focus on a short circuit for approximate implementation of the hashing algorithm (pseudo rotation approach) which was introduced in \cite{kalis18,ziiatdinov2023gaps}.
This approach allows us to exponentially reduce the number of CNOT gates, but , as computational experiments show  \cite{ziiatdinov2023gaps}, the probability of error increases slightly.   At the same time, the exponential reduction of CNOT gates gives us hope to obtain a reasonable implementation on quantum devices.

Moreover, for quantum devices like IBMQ there are some restrictions to apply CNOT gates or other two-qubit gates. Namely, we cannot apply the gate to any pair of qubits. There is a graph that shows which pairs of qubits can interact when two-qubit operators are applied. This graph represents a topology of qubits. 

When researchers pay attention for topology, they consider the linear topology called linear nearest neighbor (LNN) architecture ~\cite{saeedi2011synthesis,o2019generalized,bako2022near}. The LNN architecture for quantum hashing was considered in \cite{ksy2024}. For instance, Baidu's quantum computer or 5-qubit IBMQ quantum computer like $\mathtt{ibmq}\_\mathtt{manila}$ have such architecture.
Nevertheless, more general graph can be an issue for techniques which were applied in the LNN case. At the same time, the techniques for the general graph are based on the LNN case.  As an example, the quantum circuit for quantum hashing in general case (not the short circuit for approximate implementation that is considered in this paper) was developed in \cite{zkk2023}. It is based on implementation of uniformly-controlled gate in the LNN architecture \cite{mottonen2005decompositions}.

In this work, we consider a short circuit for approximate implementation of the hashing algorithm (pseudo rotation approach). Particularly, we use it to  implement the automaton for $\mathtt{MOD_p}$ language on real devices which are 16-qubit and 27-qubit IBMQ machines. They are Falcon r4p architecture for 16-qubit device (that is presented by $\mathtt{ibmq}\_\mathtt{guadelupe}$), and Falcon 5.11 architecture for  27-qubit device (that is presented by $\mathtt{ibmq}\_\mathtt{auckland}$ or $\mathtt{ibmq}\_\mathtt{cairo}$).
These quantum devices have specific graphs for the qubits topology  that are different from LNN architecture.
We propose an  approach for constructing efficient quantum circuits with reduced  number of CNOT gates for implementing on quantum devices with the aforementioned  topologies.  

We show that our approach results in CNOT-cost that is three times more efficient than the original approach after a standard Qiskit optimization procedure called ``transpillation''. This advantage raises with the length of the input. 

Note that the short circuit for approximate implementation of the hashing algorithm has many common moments with the circuit for Quantum Fourier Transform (QFT) algorithm. So, we apply our approach for QFT algorithm implementation. Notably, the quantum circuit for the QFT algorithm is represented by controlled phase  gates $R_k$ followed by one-qubit Hadamard operators. 

In \cite{park2023reducing}, authors propose a quantum circuit with reduced number CNOT gates implementing the QFT algorithm in LNN architecture. The CNOT-cost of their optimized circuit is $n^2+n-4$, where $n$ is the number of qubits. At the same time, their approach cannot be applied for 
non-LNN architecture directly, even for other nearest-neighbour architectures. 

In this paper, we present an optimized circuit for the QFT algorithm implementation on IBMQ quantum devices which are the 16-qubit Falcon r4p architecture  ($\mathtt{ibmq}\_\mathtt{guadelupe}$), and 27-qubit Falcon 5.11 architecture ($\mathtt{ibmq}\_\mathtt{auckland}$ or $\mathtt{ibmq}\_\mathtt{cairo}$). We obtain the circuit with $1.5n^2$ CNOT-cost.

The paper is organized as follows. Section \ref{sec:pre} presents the basic notations, definitions and the QFA algorithm for $\modp$ language. Here, we also give brief explanation of the pseudo rotation approach for hashing algorithm implementation. In Section \ref{sec:real-hardware}, we propose the optimized circuit for implementing the hashing algorithm on  machines with 16-qubit Falcon r4p  and 27-qubit Falcon 5.11 architectures. In Section \ref{sec:qft}, we present the optimized circuit for implementing the QFT algorithm on these machines. Finally, we summarise all the results and highlight some open questions in Section \ref{sec:con}.

\section{Preliminaries}
\label{sec:pre}

In this section, we give a brief review on quantum finite automata (QFAs) and present the QFA algorithm recognizing $\modp$ as an algorithm which is a base for a quantum hashing algorithms. We refer the reader to \cite{NC00} for the basics of quantum computing and to~\cite{AY21} for further details on QFAs.

\subsection{Quantum Finite Automaton} 
The QFA algorithm for $\modp$ \cite{af98} lies on a basis of quantum hashing algorithm as a kernel method computing the hash of the input word. The Moore-Crutchfield model of QFA \cite{MC00} is used for construction this algorithm; we refer to this model as QFA throughout the paper. Formally, an $n$-state QFA is a 5-tuple 
\[
	M=(\Sigma,Q,\{ U_\sigma \mid \sigma \in \Sigma \cup \{\cent,\dollar \}\},q_I,Q_a ),
\] where 

\begin{itemize}
	\item $\Sigma$ is the input alphabet not containing the left and right end-markers (resp., $\cent$ and $\dollar$),
	\item $Q = \{q_1,\ldots,q_n\}$ is the set of states,
	\item $U_\sigma \in \mathbb{C}^{n \times n}$ is the unitary transition matrix for symbol $\sigma$,
	\item $q_I \in Q$ is the initial state, and,
	\item $Q_a \subseteq Q$ is the set of accepting states.
\end{itemize}

From the set of all strings $\Sigma^*$, let $x$ be the given input with $l$ symbols, i.e, $x = x_1 x_2 \cdots x_l$. The computation of $M$ is traced by an $n$-dimensional complex-valued vector called the quantum state. At the beginning of computation, $M$ is in quantum state $\ket{v_0} = \ket{q_I} $, having zeros except $I$-th entry which is 1. The input $x$ is processed by $M$ one symbol at a time which can be traced by applying the transition matrix of each scanned symbol including the end-markers. Thus, the final quantum state is 
$
\ket{v_f} = U_\dollar U_{x_l} U _{x_{l-1}} \cdots U_{x_1} U_{\cent} \ket{v_0}.
$ After processing the whole input, a measurement in the computational basis is performed. That is, if the $j$-th entry of $\ket{v_f}$ is $a_j$, then the state $q_j$ is observed with probability $|a_j|^2$. Thus, the input $x$ is accepted with the probability of observing an accepting state, i.e., $ \sum_{q_j \in Q_a} |a_j|^2 $.

A language $L \subseteq \Sigma^*$ is recognized by $M$ with bounded error if and only if there exists an error bound $\epsilon \in [0,1/2) $ such that

\noindent
(i)  for each $x \in L$, $M$ accepts $x$ with probability at least $1 - \epsilon$ and 

\noindent
(ii) for each $x \notin L$, $M$ accepts (rejects) $x$ with probability at most $\epsilon$ (at least $1 - \epsilon$).

\subsection{QFA Algorithm for $\modp$}
\label{sec:QFAforMODp}

In \cite{af98,an2008,an2009}, it is proven that there exists a QFA with $O(\log p)$ states recognizing the language $\modp$ with bounded error. Let $\Sigma=\{a\}$, the automaton is unary one.
Before discussing the automaton, we start with a simpler construction. For a fixed  $k \in \{1,\ldots,p-1\}$, let $M_k$ be a QFA with 2 states, namely $\ket{q_1}$ and $\ket{q_2}$. The QFA $M_k$ starts its computation in $\ket{q_1}$ and applies the identity matrix when reading $\cent$ or $\dollar$. When reading symbol $a$, $M_k$ applies a rotation on the (real) plane $\ket{q_1}$-$\ket{q_2}$ with angle $k\frac{2\pi}{p}$. Here, $q_1$ is the accepting state. The QFA $M_k$ can be implemented by using a single qubit by associating $\ket{q_1}$ and $\ket{q_2}$ with states $\ket{0}$ and $\ket{1}$, respectively. Then, the rotation matrix (parameterized with $k$) is
\[
A_k = \mymatrix{cc}{\cos \mypar{\frac{2k\pi}{p}} & -\sin \mypar{\frac{2k\pi}{p}} \\ \sin \mypar{\frac{2k\pi}{p}} & \cos \mypar{\frac{2k\pi}{p}}}.
\]


Next, we describe a $2d$-state QFA $M_K$ for the $\modp$ language, where $K=\{k_1,\ldots,k_d\}$. We pick $d$ as a power of 2 for simplicity. The QFA $M_K$ executes the QFAs $M_{k_1}$, \dots, $M_{k_d}$ in parallel. After reading $\cent$, $M_K$ enters into equal superposition of $M_{k_j}$'s for $j \in \{1,\ldots,d\}$. Each $M_{k_j}$ uses two states, and when reading each $a$, $M_{k_j}$ applies the rotation matrix $A_{k_j}$. At the end of the computation, $M_K$ applies the inverse of the transition matrix used for the symbol $\cent$. It was shown in \cite{af98,AN09} that there exist $O(\log p)$ $k$ values such that $M_K$ recognizes $\modp$ with bounded error.

We can use Hadamard operators for the end-markers for implementation of $M_K$ on real hardware. Consequently, the main cost of the implementation comes from the transition matrix for the symbol $a$, which is defined as
\[
U_a = \mymatrix{c| c | c | c}{ A_{k_1} & 0 & \cdots & 0 \\ \hline 0 & A_{k_2} & \cdots & 0 \\ \hline \vdots & \vdots & \ddots & \vdots \\ \hline 0 & 0 & \cdots & A_{k_d} } .
\]
The operator $U_a$ is also known as uniformly controlled rotation \cite{mottonen2006decompositions}. Thus, we can use its efficient implementations from \cite{mottonen2006decompositions} for circuit implementation.

\subsection{Connection between QFA for $\modp$ and Hashing Algorithm}

In the QFA for $\modp$ language, transformation of the quantum register after reading all the input can be  represented as a computing of a  fingerprint or a hash of the input word. The hash represents a length of the input by modulo $p$. When one needs to compare two strings, it can be done  by comparing their hashes by modulo $p$ \cite{av2013}. In the quantum case, for the first string, the target qubit is rotated counterclockwise, for the second string, on the contrary, it is rotated clockwise. Then, if the measured result is in the state $|0\rangle$, one can assume that strings are equal with some probability. So, the quantum hashing algorithm is computationally based on the QFA algorithm for $\modp$ language, and further, by the QFA for  $\modp$ we mean a hashing algorithm.

\subsection{Pseudo Rotations Approach}
\label{sec:pseudo}
We focus on the number of CNOT gates when say about a cost of a circuit. The basic operation of the  QFA for  $\modp$ is $U_a$ which is composed of $d$ rotations, where $\log_2 d+1$ is a width of the circuit. It can be implemented by applying $d$ CNOT gates \cite{mottonen2006decompositions}.
For reducing the number of CNOT gates we use an approach presented in \cite{kalis18,ziiatdinov2023gaps}. Namely, the circuit uses $\log_2 d +1 $ rotations and apply them as shown in Fig.~\ref{fig:pseudo}. In this case, rotation angles are chosen as a linear combination of a set of angles of size $\log_2 d+1$, and so the set $K$ of size $d$ is not determined independently.  This short circuit cannot be an implementation for an arbitrary uniformly controlled rotation gate. However, numerical experiments given in \cite{ziiatdinov2023gaps}, show us that the circuit is reasonable for the quantum hashing algorthm. Authors of \cite{ziiatdinov2023gaps} presented angles that produces the required set $K$ which allows to recognize $\modp$ with bounded error at most $\frac{1}{3}$.

\begin{figure}[h]
\begin{center}
\begin{tikzpicture}
\node[scale=0.6] {
\begin{quantikz}[column sep=0.2cm, row sep=0.2cm]
            q_{1}&=\ket{0} & \gate{H} & \qw & \qw  &\qw & \ldots && \qw   & \ctrl{6}& \gate{H} &\meter{} & \qw \\
           q_{2}&=\ket{0} & \gate{H}& \qw  & \qw &\qw &\ldots && \ctrl{5}&\qw     &  \gate{H} &\meter{} & \qw \\
           \\
           & & \wave &&&&&&&&& \\
            & & \\
                q_{\log d}&= \ket{0} & \gate{H} & \qw & \ctrl{1}   &\qw & \ldots && \qw  &  \qw &\gate{H} &\meter{} & \qw \\
                q_{\log d+1}&=\ket{0} & \qw & \gate{T_{\xi_0}} &  \gate{T_{\xi_1}}  & \qw & \ldots  && \gate{T_{\xi_{\log d-1}}} &   \gate{T_{\xi_{\log d}}} &\qw & \meter{}& \qw
\end{quantikz}
};
\end{tikzpicture}
\end{center}
\caption{Pseudo rotations approach.}
\label{fig:pseudo}
\end{figure}
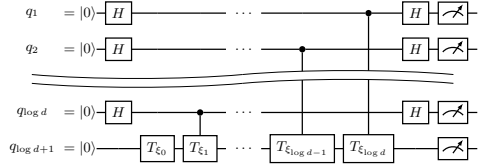


\section{Implementations on Real Quantum Devices}
\label{sec:real-hardware}

Quantum computers execute programs using a basic set of gates, and each backend works with specifically structured set of qubits. Before execution, any program is transpiled into a circuit corresponding to the basis gates and the topology of qubits. The framework  Qiskit provides the transpiler which uses optimization algorithms for reducing computational error \cite{transpiler,corcoles2021exploiting}. However, the transpillation process is generic and so may provide not the best solutions for the circuits decomposition.   

Let us note that a CNOT gate between two qubits that are not neighbor to each other would be implemented by using several CNOT gates. In this section, we show how we can reduce the CNOT-cost of the circuit based on the choice of the topology of the real hardware.  For the implementation, we use tools and libraries provided by Qiskit framework~\cite{qiskit_framework}. 
We consider such backends of IBM Quantum (IBMQ) \cite{ibmqbackends} as 16-qubit $\mathtt{ibmq}\_\mathtt{guadelupe}$ and 27-qubit $\mathtt{ibmq}\_\mathtt{cairo}$ backends. The 16-qubit device is an instance of Falcon r4p architecture and the 27-qubit device is an instance of Falcon 5.11 architecture.

Almost all IBMQ backends use the set of gates $\{CX, I, R_z, SX, X \}$ as basis gates. Here, $CX$ is the CNOT gate, $I$ is the identity gate, $X$ is the NOT gate, and $R_z$ and $SX$ are defined as:
\[
	R_z(\theta) = \mymatrix{cc}{e^{-i\theta/2} & 0 \\ 0 & e^{i \theta/2}}
~~
	SX =  \frac{1}{2} 
	\begin{pmatrix}
        1+i & ~ & 1-i \\
        1-i &~ & 1+i \\
    \end{pmatrix}.
\]
For the circuit implementation, we use lemmas proposed in \cite{sy2021} which claim the following. 
\[
	SX^{\dagger} R_z(2k\pi/p) SX = R_y(2k\pi/p),
\]
\[	(R_y(2k\pi/p))^j = (SX^{\dagger} R_z(2k\pi/p) SX )^j = SX^{\dagger} R^j_z(2k\pi/p) SX  
\]

\subsection{Optimization of Circuits with Respect to a Qubit Topology}

In this section, we give an optimization approach for implementing quantum hashing algorithm for a given quantum devices.

 At first, let us consider the  16-qubit $\mathtt{ibmq}\_\mathtt{guadelupe}$ backend which has  the topology of qubits  resembling a circle with branches (see Fig.\ref{fig:guadelupe}).

\begin{figure}[h!]
\centering
\includegraphics[width=5cm]{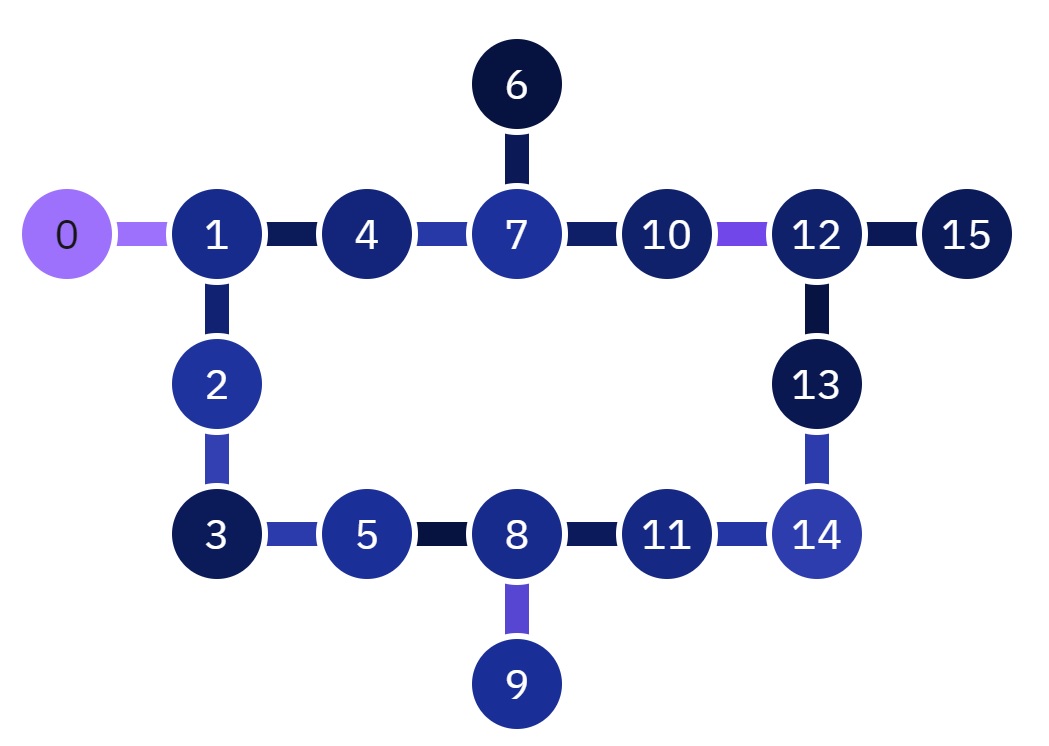}
\caption{16-qubit topology of  $\mathtt{ibmq}\_\mathtt{guadelupe}$ machine.}
\label{fig:guadelupe}
\end{figure}

We implement the hashing algorithm with pseudo rotations given in Section~\ref{sec:pseudo} using the gate equivalences given above.
For testing the approach, we pick different values of $p$ and values of the set 
$K$ according to the results given in \cite{ziiatdinov2023gaps}. We optimize a circuit and run it on the simulator of  $\mathtt{ibmq}\_\mathtt{guadelupe}$ machine using $\mathtt{FakeGuadalupeV2}$ backend.  Based on the topology, we can minimize the number of CNOT gates.


Let us look at the standard qiskit's transpiling procedure. Since controlled operations can be processed only on neighbor qubits in the real machines, a transpilator swaps qubits to make a controller qubit close to the target qubit. This also increases the number of CNOT gates in a circuit.
We restructure a circuit in the following way. For a current CNOT-gate, the target qubit is swapped with the neighboring one if the control qubit is not adjacent to it. For the topology in Fig.\ref{fig:guadelupe}, for each input symbol $x$ the target qubit is initially in a qubit $q_1$ (position 1 in a graph). For each controlled rotation gate, the target qubit goes by swapping to the 4-th, then to the 7-10-12-13-14-11-8-5 qubits. Finally, to come back to the initial position $q_1$, the target goes from the 5-th to $3$-$2$-$1^{st}$ position.

The circuit equality given in Fig. \ref{fig:crz} shows how to implement
 the controlled $R_z$ rotation ($CR_z$) using two CNOT gates and two $R_z$ gates. However, we cannot apply CNOT gate to any pair of qubits, but only to neighbor qubits. That is why we use SWAP gate that exchanges two qubits. It can be presented as three CNOT gates that is shown in the circuit equality in Fig. \ref{fig:swap}. At the same time, if we apply $CR_z$ rotation and then move the target qubit using SWAP gate, then two sequential CNOT gates cancel each other as it presented in Fig.  \ref{fig:rotation}. So, the optimized procedure of moving the target qubit followed by its' rotation costs  $3$ CNOT gates, when the simple controlled rotation of the target qubit takes $2$ CNOT gates. Note that the target qubit does not move to the positions 0,2,6,15, 9 and 3 of the topology graph (see Fig.\ref{fig:guadelupe}), so each interacting of the target with the qubits in these positions costs only  $2$ CNOT gates.

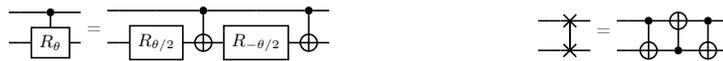
\begin{figure}[h!]
\begin{subfigure}{0.5\textwidth}
\centering
\begin{tikzpicture}
\node[scale=0.75] {
\begin{quantikz}[column sep=0.2cm, row sep=0.2cm]
     & \qw&\ctrl{1} & \qw 
    \\
    &\qw & \gate{R_{\theta}} & \qw          
\end{quantikz}=\begin{quantikz}[column sep=0.2cm, row sep=0.2cm]
     & \qw&\qw&\ctrl{1} &\qw &\ctrl{1} & \qw   
    \\
    &\qw & \gate{R_{\theta/2}}  &\targ{} &    \gate{R_{-\theta/2}}  &\targ{} & \qw     
\end{quantikz}
};
\end{tikzpicture}
\caption{Circuit equality for  a controlled rotation}.
\label{fig:crz}
\end{subfigure}
\begin{subfigure}{0.5\textwidth}
\centering
\begin{tikzpicture}
\node[scale=0.75] {
\begin{quantikz}[column sep=0.2cm, row sep=0.2cm]
     & \qw&\swap{1} & \qw 
    \\
    &\qw & \swap{-1} & \qw          
\end{quantikz}=\begin{quantikz}[column sep=0.2cm, row sep=0.2cm]
     & \qw&\ctrl{1} &\targ{}&\ctrl{1} & \qw   
    \\
    &\qw  &\targ{} & \ctrl{-1} &\targ{} & \qw     
\end{quantikz}
};
\end{tikzpicture}
\caption{Circuit equality for a SWAP gate}.
\label{fig:swap}
\end{subfigure}
\caption{Circuit equalities}.
\end{figure}

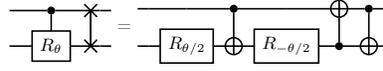
\begin{figure}[h!]
\begin{center}
\begin{tikzpicture}
\node[scale=0.75] {
\begin{quantikz}[column sep=0.2cm, row sep=0.2cm]
     & \qw&\ctrl{1} &\swap{1} & \qw 
    \\
    &\qw & \gate{R_{\theta}}  &\swap{-1} & \qw          
\end{quantikz}=\begin{quantikz}[column sep=0.2cm, row sep=0.2cm]
     & \qw&\qw&\ctrl{1} &\qw& \targ{} &\ctrl{1} & \qw   
    \\
    &\qw & \gate{R_{\theta/2}}  &\targ{} &    \gate{R_{-\theta/2}} & \ctrl{-1} &\targ{} & \qw     
\end{quantikz}
};
\end{tikzpicture}
\end{center}
\caption{Circuit equality for a SWAP gate followed by a controlled rotation gate}.
\label{fig:rotation}
\end{figure}


The part of the proposed circuit is shown as the following (see Fig.\ref{fig:pseudodec}).
We note that the number of CNOT gates does not change after using Qiskit transpiler with the maximal possible optimization level 3.



\begin{figure}[h!]
\begin{center}
\begin{tikzpicture}
\node[scale=0.55] {
\begin{quantikz}[column sep=0.2cm, row sep=0.2cm]
            q_{7}=\ket{0} & \gate{H} & \qw  &\qw & \qw  & \qw& \qw &\qw &\qw & \qw  & \qw& \qw& \qw & \qw  & \qw& \qw      & \ctrl{1}& \qw   & \targ{}& \ctrl{1}& \qw & \ldots\\  
            q_{4}=\ket{0} & \gate{H} & \qw  &\qw & \qw  & \qw& \qw & \qw & \qw  & \qw& \qw      & \ctrl{2}& \qw   & \targ{}& \ctrl{2}&\gate{R_z(\theta_4)} & \targ{} &\gate{R_z(\theta_4')}& \ctrl{-1}&\targ{}&\qw & \ldots\\   
           q_{2}=\ket{0} & \gate{H}& \qw  &\qw  & \qw &\qw  & \qw  & \ctrl{1}&\qw &  \ctrl{1}   & \qw & \qw &\qw & \qw &\qw & \qw  & \qw &\qw  &\qw &\qw  &\qw  & \ldots\\
           q_{1}=\ket{0} &\gate{S}&\gate{R_z(\theta_1)}&  \targ{} &\gate{R_z(\theta_1')} & \targ{} &\gate{R_z(\theta_2)} &\targ{} &\gate{R_z(\theta_2')}  & \targ{} &  \gate{R_z(\theta_3)} & \targ{} &\gate{R_z(\theta_3')} & \ctrl{-2} &\targ{} &\qw  & \qw &\qw  & \qw & \qw & \qw & \ldots\\
           q_{0}=\ket{0} &\gate{H} &\qw  &\ctrl{-1} &\qw &\ctrl{-1}  &  \qw  & \qw &\qw  & \qw & \qw  & \qw &\qw  & \qw &\qw  & \qw    & \qw& \qw & \qw & \qw & \qw & \ldots\\    
\end{quantikz}
};
\end{tikzpicture}
\end{center}
\caption{Proposed circuit for the input $a$ for $\mathtt{MOD_{p}}$ language}.
\label{fig:pseudodec}
\end{figure}
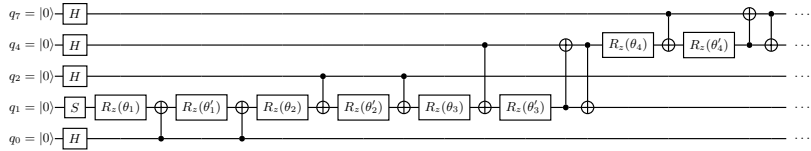

We implemented both our circuit and the pseudo circuit. For comparison, we also transpiled the original pseudo circuit with optimization level 3. The results are given in Fig.~\ref{fig:cxs} through which we can observe that we reduce the number of CNOT gates by more than three times. 

\begin{figure}[h!]

\centering
\includegraphics[width=8cm]{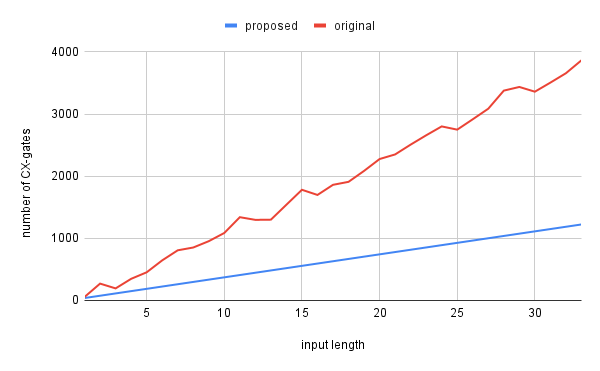}
\caption{The number of CNOT gates after transpillation for the pseudo circuit constructed using the original approach and our approach for $\modp$.}
\label{fig:cxs}
\end{figure}


For the 27-qubit graph corresponding to the topology of machines $\mathtt{ibmq}\_\mathtt{auckland}$ or $\mathtt{ibmq}\_\mathtt{cairo}$, the target qubit goes through 1-4-7-10-12-15-18-21-23-24-25-22-19-16-14-11-8-5 qubits if we can recognize parity of an index of an input symbol. The reverted chain is used for even symbols. If we cannot recognize parity of an index, and the circuit should be unified for all symbols, then the chain is  1-4-7-10-12-15-18-21-23-24-25-22-19-16-14-11-8-5-3-2. An image of such topology is shown in Fig.\ref{fig:auck}. 

\begin{figure}[h!]

\centering
\includegraphics[width=10cm]{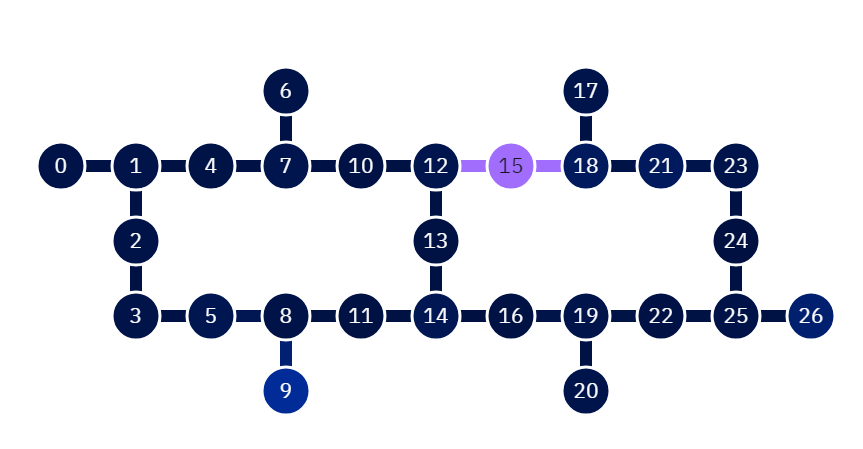}
\caption{Qubits' topology of  $\mathtt{ibmq}\_\mathtt{auckland}$, $\mathtt{ibmq}\_\mathtt{cairo}$ machines.}
\label{fig:auck}
\end{figure}

In the rest part of the paper we consider only the case when we can recognize parity of an index. Another case can be considered in the similar way.

We calculate the number of CNOT gates of our program executed on a device with such kind of topologies. For 16-qubit topology (Fig.\ref{fig:guadelupe}), there are $6$ rotations without moving a target qubit (for qubits number 0,2,6,15,9, and 3) and 9 rotations with moving the target qubit (by chain 1-4-7-10-12-13-14-11-8-5). The total number of CNOTs is $39$, whereas the  transpiled  original circuit uses $51$ CNOT gates.
For 27-qubit topology, the transpiled optimized circuit contains $69$ CNOT gates while the transpiled  original circuit has $109$ CNOT gates.

When reading more than one input symbol, we can also decrease the number of CNOT gates by joining some control rotations for the $i$-th and $(i+1)$-th symbols which is presented in \cite{ksy2024} in details. We decrease it for $2$ CNOT gates for each symbol because we have two sequential $CR_z$ gates that can be replaced by one $CR_z$ gate rotating the qubit by the sum of angles.   
Thus, the experiments  showed that the difference between costs of the proposed and original circuits increases with the length of the input. It shows that our optimization approach gives a better CNOT-cost compared to the results reached by the Qiskit transpiler. 

So, we have the following two results on CNOT-cost for a string of length $l$.

\begin{theorem}
The presented quantum circuit implements QFA for $\modp$ on an input string of length $l$. 
\begin{itemize}
    \item For 16-qubit $\mathtt{ibmq}\_\mathtt{guadelupe}$, it has $37\cdot l +4$ CNOT gates if $l\geq 3$, $39$ if $l=1$ and $76$ if $l=2$.  
    \item For 27-qubit $\mathtt{ibmq}\_\mathtt{auckland}$ or $\mathtt{ibmq}\_\mathtt{cairo}$, it has $67\cdot l +4$ CNOT gates if $l\geq 3$, $69$ if $l=1$ and $136$ if $l=2$.  
\end{itemize} 
\end{theorem}
\begin{proof}
    Let us consider 16-qubit machine. For each symbol we process $6$ control rotations without swaps (for qubits 0,2,6,15,9, and 3), each of them has $2$ CNOT gates; and $9$ control rotations with swaps (for qubits 4, 7, 10, 12, 13, 14, 11, 8, and 5), each of them has $3$ CNOT gates. Totally, we have $39$ CNOT gates for each symbol. At the same time, for each two sequential symbols we remove $2$ CNOT gates. We do not do it for $l=1$. We remove $2$ CNOT gates if $l=2$ and $2(l-2)$ CNOT gates in other cases.

    The proof for 27-qubit machine is the similar, but in that case we have  $9$ control rotations without swaps (for qubits 0, 2, 6, 13, 17, 26, 20, 9, and 3), each of them has $2$ CNOT gates; and $17$ control rotations with swaps (for qubits 4, 7, 10, 12, 15, 18, 21, 23, 24, 25, 22, 19, 16, 14, 11, 8, and 5), each of them has $3$ CNOT gates.
\end{proof}

\subsection{Other IBMQ Topologies}

To date,  different IBMQ backends  such as $\mathtt{ibmq}\_\mathtt{auckland}$, 
$\mathtt{ibmq}\_\mathtt{cairo}$, $\mathtt{ibmq}\_\mathtt{osaka}$, $\mathtt{ibmq}\_\mathtt{kioto}$, $\mathtt{ibmq}\_\mathtt{cleveland}$ etc. have the similar architectures of qubits $\it{Eagle}, \it{Falcon}, \it{Heron}$. So, hashing algorithm can be implemented on such machines in the similar way which we presented above. Namely, the target qubit moves through the following ``snake'' trace as it shown by a red line in a Fig.\ref{fig:eagle}.

\begin{figure}[h!]

\centering
\includegraphics[width=7cm]{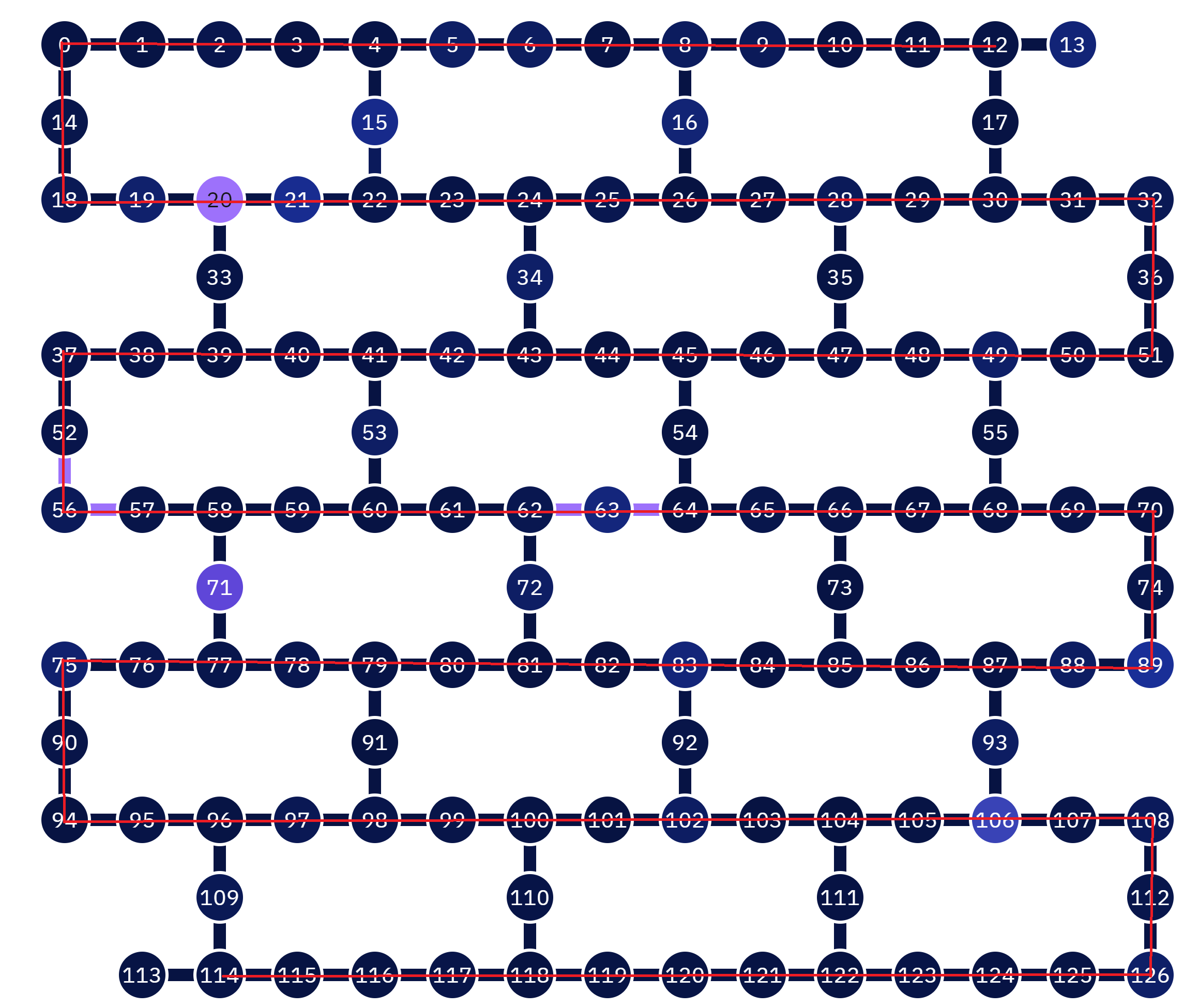}
\caption{A qubit topology of a $\mathtt{ibmq}\_\mathtt{torino}$ machine with $\it{Heron}$ architecture}
\label{fig:eagle}
\end{figure}

 \section{QFT Algorithm Implementation on Real Devices}
 \label{sec:qft}
There are also other algorithms with the similar structure of the circuit. For instance, we consider Quantum Fourier Transform algorithm.
This algorithm plays important role  and can be considered as a base of many algorithms like quantum addition \cite{d2000}, quantum phase estimation (QPE) \cite{k1995}, quantum amplitude estimation (QAE) \cite{bhmt2002}, the algorithm for solving linear systems of equations \cite{hhl2009}, quantum walks \cite{mnrs2007}, and Shor’s factoring algorithm \cite{s97}, to name a few. Therefore, the cost optimization of QFT would result in the efficiency improvement of these quantum algorithms.

In our work, we skip details of the QFT algorithm and only focus on its implementation circuit (see Fig. \ref{fig:qft}), where $R_k$ is a  controlled phase gate represented by a matrix 
$R_k=\begin{pmatrix}
 1 & 0\\
 0 & e^{i 2 \pi/2^k}
\end{pmatrix}$.
You can find  more details about the algorithm in  \cite{nc2010}.
In \cite{park2023reducing}, authors presented the circuit equality for $CR_k$ decomposition using CNOT and $R_z$ rotation gates. Add the SWAP gate we obtain the following circuit.
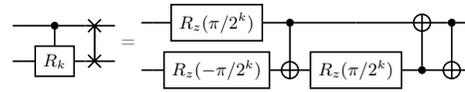
\begin{figure}[h!]
\begin{center}
\begin{tikzpicture}
\node[scale=0.75] {
\begin{quantikz}[column sep=0.2cm, row sep=0.2cm]
     & \qw&\ctrl{1} &\swap{1} & \qw 
    \\
    &\qw & \gate{R_{k}}  &\swap{-1} & \qw          
\end{quantikz}=\begin{quantikz}[column sep=0.2cm, row sep=0.2cm]
      &\qw& \gate{R_z(\pi/2^k)}&\ctrl{1} &\qw& \targ{} &\ctrl{1} & \qw   
    \\
    &\qw & \gate{R_z(-\pi/2^k)}  &\targ{} &    \gate{R_z(\pi/2^k)} & \ctrl{-1} &\targ{} & \qw     
\end{quantikz}
};
\end{tikzpicture}
\end{center}
\caption{Circuit equality for a SWAP gate followed by a controlled phase gate}.
\label{fig:rotation}
\end{figure}

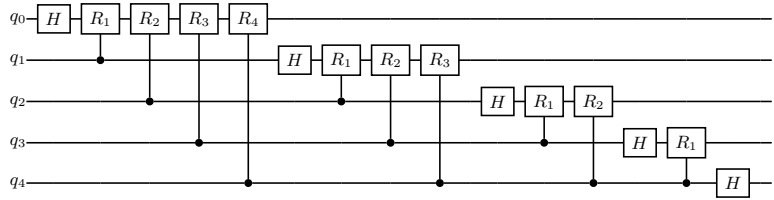
\begin{figure}[h!]
\begin{center}
\begin{tikzpicture}
\node[scale=0.75] {
\begin{quantikz}[column sep=0.2cm, row sep=0.2cm]
            q_{0}& \gate{H} & \gate{R_1}  &\gate{R_2}& \gate{R_3} & \gate{R_4}& \qw &\qw &\qw & \qw  & \qw& \qw& \qw & \qw  & \qw& \qw &\qw & \qw & \qw      \\
            q_{1} & \qw  &\ctrl{-1} & \qw  & \qw& \qw & \gate{H} & \gate{R_1}  &\gate{R_2}& \gate{R_3} & \qw& \qw &\qw & \qw &\qw & \qw &\qw & \qw & \qw  \\    
           q_{2} & \qw& \qw  &\ctrl{-2} & \qw &\qw  & \qw  & \ctrl{-1}&\qw &  \qw   & \qw&\gate{H} & \gate{R_1}  &\gate{R_2}&  \qw &\qw & \qw &\qw & \qw    \\
           q_{3} & \qw& \qw  &\qw&\ctrl{-3} & \qw &\qw & \qw & \ctrl{-2} & \qw & \qw &\qw &\ctrl{-1}  & \qw & \gate{H} & \gate{R_1} & \qw & \qw & \qw & \\
           q_{4}  & \qw& \qw  &\qw&\qw &\ctrl{-4} &\qw&  \qw  & \qw  & \ctrl{-3} & \qw  & \qw & \qw &\ctrl{-2}  & \qw &\ctrl{-1}  & \gate{H}    & \qw& \qw \\   
           
\end{quantikz}
};
\end{tikzpicture}
\end{center}
\caption{ Conventional 5-qubit quantum Fourier transform circuit}.
\label{fig:qft}
\end{figure}

As one can note, it contains a cascades of controlled phase gates similar to the cascade of controlled rotation gates in the circuit for hashing in Fig.\ref{fig:pseudo}, where a target qubit is also fixed for each such cascade. 

So, we can claim that our circuit optimization approach works also for the QFT algorithm implementation. 

Namely, if the qubits are linearly structured, then the target qubit moves to the controller by swapping with a neighbor qubit for processing each controlled rotation.

We apply this approach of the circuit construction for QFT implementing on 16-qubit $\mathtt{ibmq}\_\mathtt{guadelupe}$ devise (see Fig.\ref{fig:guadelupe}) and 27-qubit $\mathtt{ibmq}\_\mathtt{cairo}$ devise (see Fig.\ref{fig:auck}). 

In Fig.\ref{fig:qft}, \ref{fig:qft2}, we present a 7-qubit circuit, where initially working qubits are set in positions 0-1-2-4-7-6-10 of the graph in Fig.\ref{fig:guadelupe}. Initially, the target qubit is set in the position 1; and after a cascade of controlled rotations, this qubit moves to the farthest position from the position 1 due to it is no longer ``useful'' in the circuit. 

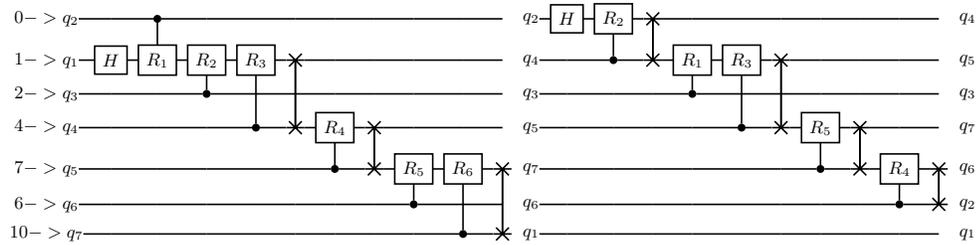
\begin{figure}[h!]
\begin{center}
\begin{tikzpicture}
\node[scale=0.75] {
\begin{quantikz}[column sep=0.2cm, row sep=0.2cm]
    0 -> q_{2} & \qw  & \ctrl{1} &\qw & \qw & \qw& \qw & \qw & \qw & \qw& \qw  & q_2  & \gate{H} & \gate{R_2}  &\swap{1}&\qw & \qw & \qw& \qw & \qw & \qw  &\qw & q_4 
    \\
    1 -> q_{1} &\gate{H} & \gate{R_1}  &\gate{R_2}& \gate{R_3} & \swap{2}&\qw &\qw &\qw & \qw  & \qw& q_4 & \qw & \ctrl{-1}  &\swap{-1}&\gate{R_1}& \gate{R_3} & \swap{2}& \qw &\qw &\qw & \qw  &  q_5
    \\
    2-> q_{3} & \qw& \qw  &\ctrl{-1} & \qw &\qw  & \qw  & \qw &  \qw & \qw & \qw  & q_3 &  \qw &\qw  &\qw&\ctrl{-1}& \qw &\qw  & \qw  & \qw &  \qw & \qw & q_3
    \\
    4->q_{4} & \qw& \qw  &\qw&\ctrl{-2} & \swap{-2} & \gate{R_4}&\swap{1}&\qw & \qw   & \qw & q_5 & \qw & \qw &\qw & \qw & \ctrl{-2} & \swap{-2} & \gate{R_5}&\swap{1}&\qw & \qw   &  q_7
    \\
    7-> q_{5}  & \qw& \qw  &\qw&\qw  &\qw&  \ctrl{-1}  & \swap{-1}&\gate{R_5} &\gate{R_6}& \swap{2} & q_7& \qw  &\qw&\qw  &\qw& \qw&  \qw&  \ctrl{-1}  & \swap{-1}&\gate{R_4} &\swap{1}& q_6
    \\   
    6->  q_{6}  & \qw & \qw  &\qw&\qw& \qw& \qw  &\qw&\ctrl{-1} &\qw&  \qw   & q_6 & \qw & \qw& \qw  &\qw&\qw& \qw& \qw  &\qw&\ctrl{-1} &  \swap{-1}  & q_2 
    \\ 
    10->  q_{7}  & \qw& \qw& \qw& \qw  &\qw&\qw  &\qw&\qw &\ctrl{-2}& \swap{-2}  & q_1& \qw  &\qw&\qw& \qw& \qw  &\qw&\qw&\qw&\qw&  \qw &q_1\\ 
           
\end{quantikz}
};
\end{tikzpicture}
\end{center}
\caption{ 7-qubit quantum Fourier transform circuit for implementing on $\mathtt{ibmq}\_\mathtt{guadelupe}$ machine. Part 1}.
\label{fig:qft}
\end{figure}

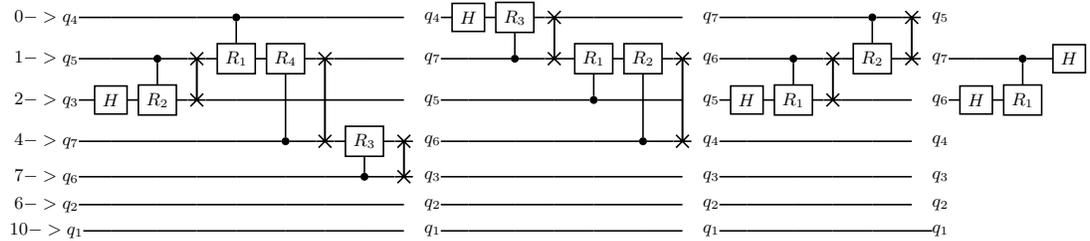
\begin{figure}[h!]
\begin{center}
\begin{tikzpicture}
\node[scale=0.75] {
\begin{quantikz}[column sep=0.2cm, row sep=0.2cm]
    0 -> q_{4} & \qw &\qw &\qw& \ctrl{1} &\qw & \qw & \qw& \qw &  q_4&\gate{H}  &\gate{R_3}&\swap{1} & \qw & \qw& \qw & q_7 &\qw & \qw& \qw &\ctrl{1}&\swap{1}& q_5 
    \\
    1 -> q_{5} & \qw & \ctrl{1} &\swap{1}&\gate{R_1}& \gate{R_4} & \swap{2}& \qw &\qw &  q_7   &\qw &\ctrl{-1}&\swap{-1}& \gate{R_1}  &\gate{R_2}& \swap{2} &  q_6 &\qw &\ctrl{1} &\swap{1}& \gate{R_2}&\swap{-1}& q_7 & \qw &\ctrl{1} &\gate{H}
    \\
    2-> q_{3} &\gate{H} & \gate{R_2}&\swap{-1} &\qw &\qw & \qw &\qw  & \qw  &q_5  & \qw  & \qw & \qw&\ctrl{-1}& \qw &\qw & q_5 &
    \gate{H} & \gate{R_1} &\swap{-1} &\qw &\qw &q_6 & \gate{H} & \gate{R_1}
    \\
    4->q_{7}& \qw& \qw  &\qw &\qw&\ctrl{-2} & \swap{-2} & \gate{R_3}&\swap{1}  & q_6 &\qw & \qw &\qw & \qw & \ctrl{-2} & \swap{-2} &  q_4 &\qw &\qw & \qw &\qw  & \qw &q_4
    \\
    7-> q_{6} & \qw  &\qw&\qw  &\qw& \qw& \qw&\ctrl{-1}  & \swap{-1}&q_3&\qw& \qw  &\qw&\qw  &\qw&  \qw& q_3 &\qw &\qw & \qw &\qw  & \qw &q_3
    \\   
    6->  q_{2}   & \qw  &\qw&\qw& \qw& \qw  &\qw&\qw&\qw &q_2&  \qw & \qw & \qw  &\qw&\qw& \qw& q_2  &\qw &\qw & \qw &\qw  & \qw &q_2 
    \\ 
    10->  q_{1}  & \qw& \qw& \qw  &\qw&\qw  &\qw&\qw &\qw& q_1& \qw & \qw  &\qw&\qw& \qw& \qw  &q_1&\qw &\qw & \qw &\qw  & \qw& \qw q_1\\ 
           
\end{quantikz}
};
\end{tikzpicture}
\end{center}
\caption{ 7-qubit quantum Fourier transform circuit for implementing on $\mathtt{ibmq}\_\mathtt{guadelupe}$ machine. Part 2}.
\label{fig:qft2}
\end{figure}

Due to the circuit equality given in Fig.\ref{fig:rotation}, each such gates structure applied to a neighboring pair of qubits costs $3$ CNOT gates.


So, we can claim that the CNOT-cost of the QFT algorithm circuit for implementing on $\mathtt{ibmq}\_\mathtt{guadelupe}$ machine is $O(\frac{3}{2} n^2)$, where $n$ is the number of working qubits in the circuit.

Generalizing the QFT algorithm for implementing on 16 qubits, we order qubits in the following way. The following Table.~\ref{tab:qft16} represents where qubits move after each cascade of controlled rotations. We skip the detailed circuit and present only positions of qubits in a graph Fig.~\ref{fig:guadelupe} for each cascade of controlled rotations. 

\begin{table}[!h] 
\centering
\caption{Repositioning of qubits $q_1$-$q_{16}$ while implementing a cirquit for QFT on 16-qubit Falcon r4p architecture }
\begin{tabular}{|r|l|l|l|l|l|l|l|l|l|l|l|l|l|l|l|l|}
\hline
Rotation cascade & \multicolumn{16}{c|}{Positions of qubits in the graph}\\
\hline
 number   & 0 & 1&2& 3 & 4& 5 & 6&7& 8 & 9& 10 & 11&12& 13 & 14& 15\\
\hline
1 & $q_2$ & $q_1$ & $q_3$ & $q_{13}$ & $q_4$ & $q_{12}$ & $q_{16}$ & $q_5 $ & $q_{11}$ & $q_{14}$ & $q_6$ & $q_{10}$ & $q_7$ & $q_8$ & $q_9$ & $q_{15}$\\

2 & $q_2$ & \color{red}{$q_4$} & $q_3$ & \color{red}{$q_{1}$} & \color{red}{$q_5$} & \color{red}{$q_{13}$} & $q_{16}$ & \color{red}{$q_6 $} & \color{red}{$q_{12}$} & $q_{14}$ & \color{red}{$q_7$} & \color{red}{$q_{11}$} & \color{red}{$q_8$ }& \color{red}{$q_9$} & \color{red}{$q_{10}$} & $q_{15}$\\

3 & \color{red}{$q_4$} & \color{red}{$q_5$} & $q_3$ & $q_{1}$ & \color{red}{$q_6$} & \color{red}{$q_{2}$} & $q_{16}$ & \color{red}{$q_7 $} & \color{red}{$q_{13}$} & $q_{14}$ & \color{red}{$q_8$} & \color{red}{$q_{12}$} & \color{red}{$q_9$} & \color{red}{$q_{10}$} & \color{red}{$q_{11}$} & $q_{15}$\\

4 & $q_4$ & \color{red}{$q_6$} & \color{red}{$q_5$} & $q_{1}$ & \color{red}{$q_7$} & $q_{2}$ & $q_{16}$ & \color{red}{$q_8 $} & \color{red}{$q_{14}$} & \color{red}{$q_{3}$} & \color{red}{$q_{9}$} & \color{red}{$q_{13}$ }& \color{red}{$q_{10}$} & \color{red}{$q_{11}$} & \color{red}{$q_{12}$} & $q_{15}$\\

5 & \color{red}{$q_6$} & \color{red}{$q_7$} & $q_5$ & $q_{1}$ & \color{red}{$q_8$} & $q_{2}$ & $q_{16}$ & \color{red}{$q_9 $} & \color{red}{$q_{4}$} & $q_{3}$ & \color{red}{$q_{10}$} & \color{red}{$q_{14}$} & \color{red}{$q_{11}$} & \color{red}{$q_{12}$} &\color{red}{ $q_{13}$} & $q_{15}$\\

6 & $q_6$ & \color{red}{$q_8$} & \color{red}{$q_7$} & $q_{1}$ & \color{red}{$q_9$} & $q_{2}$ & $q_{16}$ & \color{red}{$q_{10}$} & $q_{4}$ & $q_{3}$ & \color{red}{$q_{11}$} & \color{red}{$q_{5}$} & \color{red}{$q_{12}$ }& \color{red}{$q_{13}$} &\color{red}{ $q_{14}$} & $q_{15}$\\

7 & \color{red}{$q_8$} &\color{red}{ $q_9$} & $q_7$ & $q_{1}$ & \color{red}{$q_{10}$} & $q_{2}$ & $q_{16}$ & \color{red}{$q_{11}$} & $q_{4}$ & $q_{3}$ &\color{red}{ $q_{12}$} & $q_{5}$ & \color{red}{$q_{13}$} &\color{red}{ $q_{14}$} & \color{red}{$q_{6}$} & $q_{15}$\\

8 & $q_8$ &\color{red}{ $q_{10}$} &\color{red}{ $q_9$} & $q_{1}$ &\color{red}{ $q_{11}$} & $q_{2}$ & $q_{16}$ &\color{red}{ $q_{12}$} & $q_{4}$ & $q_{3}$ & \color{red}{$q_{13}$} & $q_{5}$ & \color{red}{$q_{14}$} & \color{red}{$q_{7}$} & $q_{6}$ & $q_{15}$\\

9& \color{red}{$q_{10}$ } &\color{red}{ $q_{11}$} & $q_9$ & $q_{1}$ &\color{red}{ $q_{12}$ }& $q_{2}$ & $q_{16}$ &\color{red}{ $q_{13}$} & $q_{4}$ & $q_{3}$ & \color{red}{$q_{14}$ }& $q_{5}$ & \color{red}{$q_{15}$} & $q_{7}$ & $q_{6}$ & \color{red}{$q_{8}$}\\

10 & $q_{10}$ &\color{red}{ $q_{12}$} & \color{red}{$q_{11}$ }& $q_{1}$ &\color{red}{ $q_{13}$ }& $q_{2}$ & $q_{16}$ & \color{red}{$q_{14}$ }& $q_{4}$ & $q_{3}$ & \color{red}{$q_{15}$} & $q_{5}$ & \color{red}{$q_{9}$} & $q_{7}$ & $q_{6}$ & $q_{8}$\\

11 &\color{red}{ $q_{12}$ }& \color{red}{$q_{13}$ }& $q_{11}$ & $q_{1}$ &\color{red}{ $q_{14}$} & $q_{2}$ & $q_{16}$ &\color{red}{ $q_{15}$} & $q_{4}$ & $q_{3}$ & \color{red}{$q_{10}$} & $q_{5}$ & $q_{9}$ & $q_{7}$ & $q_{6}$ & $q_{8}$\\

12 & $q_{12}$ & \color{red}{$q_{14}$ }& \color{red}{$q_{13}$} & $q_{1}$ &\color{red}{ $q_{15}$ }& $q_{2}$ &\color{red}{ $q_{11}$} & \color{red}{$q_{16}$} & $q_{4}$ & $q_{3}$ & $q_{10}$ & $q_{5}$ & $q_{9}$ & $q_{7}$ & $q_{6}$ & $q_{8}$\\

13 &\color{red}{ $q_{14}$} &\color{red}{ $q_{15}$} & $q_{13}$ & $q_{1}$ &\color{red}{ $q_{16}$} & $q_{2}$ & $q_{11}$ & \color{red}{$q_{12}$} & $q_{4}$ & $q_{3}$ & $q_{10}$ & $q_{5}$ & $q_{9}$ & $q_{7}$ & $q_{6}$ & $q_{8}$\\

14 &\color{red}{ $q_{13}$} & \color{red}{$q_{14}$ }& \color{red}{$q_{15}$} & $q_{1}$ & $q_{16}$ & $q_{2}$ & $q_{11}$ & $q_{12}$ & $q_{4}$ & $q_{3}$ & $q_{10}$ & $q_{5}$ & $q_{9}$ & $q_{7}$ & $q_{6}$ & $q_{8}$\\

15 & $q_{13}$ &\color{red}{ $q_{15}$} & \color{red}{$q_{14}$} & $q_{1}$ & $q_{16}$ & $q_{2}$ & $q_{11}$ & $q_{12}$ & $q_{4}$ & $q_{3}$ & $q_{10}$ & $q_{5}$ & $q_{9}$ & $q_{7}$ & $q_{6}$ & $q_{8}$\\
\hline
\end{tabular}
\label{tab:qft16}
\end{table}

Qubits that are moved while processing a rotation cascade are marked  by the red color.  For instance, while the first cascade is running, $11$ qubits are moved. Therefore, there are $10$ rotations with swaps and $5$ rotations without qubit moving are executed. Total CNOT-cost of the first cascade is $10 \cdot 3 + 5 \cdot 2=40$. The second cascade excludes interactions with the qubit $q_1$. Its CNOT-cost is $9 \cdot 3 + 5 \cdot 2=37$ In the same way we can easily compute the CNOT-cost of the whole circuit analysing each row in Table.~\ref{tab:qft16} and summarising CNOT-costs of each cascade implementing. This CNOT-cost is equal to $324$. 

 Moreover, we provide Table.~\ref{tab:qft27} that represents the qubits repositioning in a case of implementation on 27-qubit architecture, for instance, on the machine $\mathtt{ibmq}\_\mathtt{guadelupe}$.

\begin{table}[!h] 
\centering
\caption{Repositioning of qubits $q_1$-$q_{27}$ while implementing a cirquit for QFT on 27-qubit Falcon 5.11 architecture }
\begin{tabular}{|l|l|l|l|l|l|l|l|l|l|l|l|l|l|l|l|l|l|l|l|l|l|l|l|l|l|l|l|}
\hline
Cascade & \multicolumn{27}{c|}{Positions of qubits in the graph}\\
\hline
 number   & 0 & 1&2& 3 & 4& 5 & 6&7& 8 & 9& 10 & 11&12& 13 & 14& 15& 16 & 17&18& 19 & 20& 21 & 22&23& 24 & 25& 26\\
\hline
1 & \color{black}{$q_2$} & \color{black}{$q_1$} & \color{black}{$q_3$} & \color{black}{$q_{21}$} & \color{black}{$q_{4}$} & \color{black}{$q_{20}$} & \color{black}{$q_{27}$} & \color{black}{$q_{5}$} & $q_{19}$ & $q_{22}$ & $q_{6}$ & $q_{18}$ & $q_{7} $ & $q_{26}$ & $q_{17}$ & $q_{8}$ & $q_{16}$ & $q_{25}$ & $q_{9}$ & $q_{15}$ & $q_{23}$ & $q_{10}$ & $q_{14}$ & $q_{11}$ & $q_{12}$ & $q_{13}$ & $q_{24}$ \\

2 & \color{black}{$q_2$} 
& \color{red}{$q_4$} 
& \color{black}{$q_3$} 
&\color{red}{$q_{1}$} 
& \color{red}{$q_{5}$} 
& \color{red}{$q_{21}$} 
& $q_{27}$ 
& \color{red}{$q_{6}$} 
& \color{red}{$q_{20}$} 
& \color{red}{$q_{22}$} 
& \color{red}{$q_{7}$} 
& \color{red}{$q_{19}$} 
& \color{red}{$q_{8} $} 
& $q_{26}$ 
& \color{red}{$q_{18}$} 
& \color{red}{$q_{9}$} 
& \color{red}{$q_{17}$} 
& $q_{25}$
 &\color{red}{ $q_{10}$} 
 & \color{red}{$q_{16}$} 
 & $q_{23}$ 
 & \color{red}{$q_{11}$} 
 & \color{red}{$q_{15}$} 
 & \color{red}{$q_{12}$} 
 & \color{red}{$q_{13}$} 
 & \color{red}{$q_{14}$} & $q_{24}$ \\

3 & \color{red}{$q_4$} & \color{red}{$q_5$} & \color{black}{$q_3$} &\color{black}{$q_{1}$} & \color{red}{$q_{6}$} & \color{red}{$q_{2}$} & $q_{27}$ & \color{red}{$q_{7}$} & \color{red}{$q_{21}$} & $q_{22}$ & \color{red}{$q_{8}$} & \color{red}{$q_{20}$} & \color{red}{$q_{9} $} & $q_{26}$ & \color{red}{$q_{19}$} & \color{red}{$q_{10}$} & \color{red}{$q_{18}$} & $q_{25}$ & \color{red}{$q_{11}$} & \color{red}{$q_{17}$} & $q_{23}$ & \color{red}{$q_{12}$} & \color{red}{$q_{16}$} & \color{red}{$q_{13}$} & \color{red}{$q_{14}$} & \color{red}{$q_{15}$} & $q_{24}$ \\

4 & \color{black}{$q_4$} & \color{red}{$q_6$} & \color{red}{$q_5$} &\color{black}{$q_{1}$} & \color{red}{$q_{7}$} & $q_{2}$ & $q_{27}$ & \color{red}{$q_{8}$} & \color{red}{$q_{22}$} & \color{red}{$q_{3}$} & \color{red}{$q_{9}$} & \color{red}{$q_{21}$} & \color{red}{$q_{10} $} & $q_{26}$ & \color{red}{$q_{20}$} & \color{red}{$q_{11}$} & \color{red}{$q_{19}$} & $q_{25}$ & \color{red}{$q_{12}$ }& \color{red}{$q_{18}$} & $q_{23}$ & \color{red}{$q_{13}$ }& \color{red}{$q_{17}$} & \color{red}{$q_{14}$} & \color{red}{$q_{15}$} & \color{red}{$q_{16}$} & $q_{24}$ \\

5 & \color{red}{$q_6$} & \color{red}{$q_7$} & \color{black}{$q_5$} &\color{black}{$q_{1}$}& \color{red}{$q_{8}$} & $q_{2}$ & $q_{27}$ & \color{red}{$q_{9}$} & \color{red}{$q_{4}$} & $q_{3}$ & \color{red}{$q_{10}$} & \color{red}{$q_{22}$} & \color{red}{$q_{11} $} & $q_{26}$ & \color{red}{$q_{21}$} & \color{red}{$q_{12}$} & \color{red}{$q_{20}$} & $q_{25}$ & \color{red}{$q_{13}$} & \color{red}{$q_{19}$} & $q_{23}$ & \color{red}{$q_{14}$} & \color{red}{$q_{18}$} & \color{red}{$q_{15}$} & \color{red}{$q_{16}$} & \color{red}{$q_{17}$} & $q_{24}$ \\

6 & \color{black}{$q_6$} & \color{red}{$q_8$} & \color{red}{$q_7$} &\color{black}{$q_{1}$}& \color{red}{$q_{9}$} & $q_{2}$ & $q_{27}$ & \color{red}{$q_{10}$} & $q_{4}$ & $q_{3}$ & \color{red}{$q_{11}$} & \color{red}{$q_{5}$} & \color{red}{$q_{12} $} & $q_{26}$ & \color{red}{$q_{22}$} & \color{red}{$q_{13}$} & \color{red}{$q_{21}$} & $q_{25}$ & \color{red}{$q_{14}$} & \color{red}{$q_{20}$} & $q_{23}$ & \color{red}{$q_{15}$} & \color{red}{$q_{19}$} & \color{red}{$q_{16}$} & \color{red}{$q_{17}$} & \color{red}{$q_{18}$} & $q_{24}$ \\

7 & \color{red}{$q_8$} & \color{red}{$q_9$} & \color{black}{$q_7$} &\color{black}{$q_{1}$} & \color{red}{$q_{10}$} & $q_{2}$ & $q_{27}$ & \color{red}{$q_{11}$} & $q_{4}$ & $q_{3}$ & \color{red}{$q_{12}$} & $q_{5}$ & \color{red}{$q_{13} $} & $q_{26}$ & \color{red}{$q_{6}$} & \color{red}{$q_{14}$} & \color{red}{$q_{22}$} & $q_{25}$ & \color{red}{$q_{15}$} & \color{red}{$q_{21}$} & $q_{23}$ & \color{red}{$q_{16}$} & \color{red}{$q_{20}$ }& \color{red}{$q_{17}$} & \color{red}{$q_{18}$} & \color{red}{$q_{19}$} & $q_{24}$ \\

8 & \color{black}{$q_8$} & \color{red}{$q_{10}$} & \color{red}{$q_9$} &\color{black}{$q_{1}$} & \color{red}{$q_{11}$} & $q_{2}$ & $q_{27}$ & \color{red}{$q_{12}$} & $q_{4}$ & $q_{3}$ & \color{red}{$q_{13}$} & $q_{5}$ & \color{red}{$q_{14} $} & $q_{26}$ & $q_{6}$ & \color{red}{$q_{15}$} & \color{red}{$q_{7}$} & $q_{25}$ & \color{red}{$q_{16}$} & \color{red}{$q_{22}$} & $q_{23}$ & \color{red}{$q_{17}$} & \color{red}{$q_{21}$} & \color{red}{$q_{18}$} & \color{red}{$q_{19}$} & \color{red}{$q_{20}$ }& $q_{24}$ \\

9 & \color{red}{$q_{10}$} & \color{red}{$q_{11}$} & \color{black}{$q_{9}$} &\color{black}{$q_{1}$} & \color{red}{$q_{12}$} & $q_{2}$ & $q_{27}$ & \color{red}{$q_{13}$} & $q_{4}$ & $q_{3}$ & \color{red}{$q_{14}$} & $q_{5}$ & \color{red}{$q_{15} $} & $q_{26}$ & $q_{6}$ & \color{red}{$q_{16}$} & $q_{7}$ & $q_{25}$ & \color{red}{$q_{17}$} & \color{red}{$q_{23}$} & \color{red}{$q_{8}$} & \color{red}{$q_{18}$} & \color{red}{$q_{22}$} & \color{red}{$q_{19}$} & \color{red}{$q_{20}$} & \color{red}{$q_{21}$} & $q_{24}$ \\

10 & \color{black}{$q_{10}$} & \color{red}{$q_{12}$} & \color{red}{$q_{11}$} &\color{black}{$q_{1}$} & \color{red}{$q_{13}$} & $q_{2}$ & $q_{27}$ & \color{red}{$q_{14}$} & $q_{4}$ & $q_{3}$ & \color{red}{$q_{15}$} & $q_{5}$ & \color{red}{$q_{16} $} & $q_{26}$ & $q_{6}$ & \color{red}{$q_{17}$} & $q_{7}$ & $q_{25}$ & \color{red}{$q_{18}$} & \color{red}{$q_{9}$} & $q_{8}$ & \color{red}{$q_{19}$} & \color{red}{$q_{23}$} & \color{red}{$q_{20}$} & \color{red}{$q_{21}$} & \color{red}{$q_{22}$} & $q_{24}$ \\

11 & \color{red}{$q_{12}$} & \color{red}{$q_{13}$} & \color{black}{$q_{11}$} &\color{black}{$q_{1}$}& \color{red}{$q_{14}$} & $q_{2}$ & $q_{27}$ & \color{red}{$q_{15}$} & $q_{4}$ & $q_{3}$ & \color{red}{$q_{16}$} & $q_{5}$ & \color{red}{$q_{17} $} & $q_{26}$ & $q_{6}$ & \color{red}{$q_{18}$} & $q_{7}$ & $q_{25}$ & \color{red}{$q_{19}$} & $q_{9}$ & $q_{8}$ & \color{red}{$q_{20}$} & \color{red}{$q_{10}$} & \color{red}{$q_{21}$} & \color{red}{$q_{22}$} & \color{red}{$q_{23}$} & $q_{24}$ \\

12 & \color{black}{$q_{12}$} & \color{red}{$q_{14}$} & \color{red}{$q_{13}$} &\color{black}{$q_{1}$} & \color{red}{$q_{15}$ }& $q_{2}$ & $q_{27}$ & \color{red}{$q_{16}$} & $q_{4}$ & $q_{3}$ & \color{red}{$q_{17}$} & $q_{5}$ & \color{red}{$q_{18} $} & $q_{26}$ & $q_{6}$ & \color{red}{$q_{19}$} & $q_{7}$ & $q_{25}$ & \color{red}{$q_{20}$} & $q_{9}$ & $q_{8}$ & \color{red}{$q_{21}$} & $q_{10}$ & \color{red}{$q_{22}$} & \color{red}{$q_{23}$} & \color{red}{$q_{24}$} & \color{red}{$q_{11}$} \\

13 & \color{red}{$q_{14}$} & \color{red}{$q_{15}$} & \color{black}{$q_{13}$} &\color{black}{$q_{1}$} & \color{red}{$q_{16}$} & $q_{2}$ & $q_{27}$ & \color{red}{$q_{17}$} & $q_{4}$ & $q_{3}$ & \color{red}{$q_{18}$} & $q_{5}$ & \color{red}{$q_{19} $} & $q_{26}$ & $q_{6}$ & \color{red}{$q_{20}$} & $q_{7}$ & $q_{25}$ & \color{red}{$q_{21}$} & $q_{9}$ & $q_{8}$ & \color{red}{$q_{22}$} & $q_{10}$ & \color{red}{$q_{23}$} & \color{red}{$q_{24}$} & \color{red}{$q_{12}$} & $q_{11}$ \\

14 & \color{black}{$q_{14}$} & \color{red}{$q_{16}$} & \color{red}{$q_{15}$} &\color{black}{$q_{1}$} & \color{red}{$q_{17}$} & $q_{2}$ & $q_{27}$ & \color{red}{$q_{18}$} & $q_{4}$ & $q_{3}$ & \color{red}{$q_{19}$} & $q_{5}$ & \color{red}{$q_{20} $} & $q_{26}$ & $q_{6}$ & \color{red}{$q_{21}$} & $q_{7}$ & $q_{25}$ & \color{red}{$q_{22}$} & $q_{9}$ & $q_{8}$ & \color{red}{$q_{23}$} & $q_{10}$ & \color{red}{$q_{24}$} & \color{red}{$q_{13}$} & $q_{12}$ & $q_{11}$ \\

15 & \color{red}{$q_{16}$} & \color{red}{$q_{17}$} & \color{black}{$q_{15}$} &\color{black}{$q_{1}$}& \color{red}{$q_{18}$} & $q_{2}$ & $q_{27}$ & \color{red}{$q_{19}$} & $q_{4}$ & $q_{3}$ & \color{red}{$q_{20}$} & $q_{5}$ & \color{red}{$q_{21} $} & $q_{26}$ & $q_{6}$ & \color{red}{$q_{22}$} & $q_{7}$ & $q_{25}$ & \color{red}{$q_{23}$} & $q_{9}$ & $q_{8}$ & \color{red}{$q_{24}$} & $q_{10}$ & \color{red}{$q_{14}$} & $q_{13}$ & $q_{12}$ & $q_{11}$ \\

16 & \color{black}{$q_{16}$} & \color{red}{$q_{18}$} & \color{red}{$q_{17}$} &\color{black}{$q_{1}$} & \color{red}{$q_{19}$} & $q_{2}$ & $q_{27}$ & \color{red}{$q_{20}$} & $q_{4}$ & $q_{3}$ & \color{red}{$q_{21}$} & $q_{5}$ & \color{red}{$q_{22} $} & $q_{26}$ & $q_{6}$ & \color{red}{$q_{23}$} & $q_{7}$ & $q_{25}$ & \color{red}{$q_{24}$} & $q_{9}$ & $q_{8}$ & \color{red}{$q_{15}$} & $q_{10}$ & $q_{14}$ & $q_{13}$ & $q_{12}$ & $q_{11}$ \\

17 & \color{red}{$q_{18}$} & \color{red}{$q_{19}$} & \color{black}{$q_{17}$} &\color{black}{$q_{1}$} & \color{red}{$q_{20}$} & $q_{2}$ & $q_{27}$ & \color{red}{$q_{21}$} & $q_{4}$ & $q_{3}$ & \color{red}{$q_{22}$} & $q_{5}$ & \color{red}{$q_{23} $} & $q_{26}$ & $q_{6}$ & \color{red}{$q_{24}$} & $q_{7}$ & \color{red}{$q_{16}$} & \color{red}{$q_{25}$} & $q_{9}$ & $q_{8}$ & $q_{15}$ & $q_{10}$ & $q_{14}$ & $q_{13}$ & $q_{12}$ & $q_{11}$ \\

18 & \color{black}{$q_{18}$} & \color{red}{$q_{20}$} & \color{red}{$q_{19}$} &\color{black}{$q_{1}$} & \color{red}{$q_{21}$} & $q_{2}$ & $q_{27}$ & \color{red}{$q_{22}$} & $q_{4}$ & $q_{3}$ & \color{red}{$q_{23}$} & $q_{5}$ & \color{red}{$q_{24} $} & $q_{26}$ & $q_{6}$ & \color{red}{$q_{25}$} & $q_{7}$ & $q_{16}$ & \color{red}{$q_{17}$} & $q_{9}$ & $q_{8}$ & $q_{15}$ & $q_{10}$ & $q_{14}$ & $q_{13}$ & $q_{12}$ & $q_{11}$ \\

19 & \color{red}{$q_{20}$} & \color{red}{$q_{21}$} & \color{black}{$q_{19}$} &\color{black}{$q_{1}$}& \color{red}{$q_{22}$} & $q_{2}$ & $q_{27}$ & \color{red}{$q_{23}$} & $q_{4}$ & $q_{3}$ & \color{red}{$q_{24}$} & $q_{5}$ & \color{red}{$q_{25} $} & $q_{26}$ & $q_{6}$ & \color{red}{$q_{18}$} & $q_{7}$ & $q_{16}$ & $q_{17}$ & $q_{9}$ & $q_{8}$ & $q_{15}$ & $q_{10}$ & $q_{14}$ & $q_{13}$ & $q_{12}$ & $q_{11}$ \\

20 & \color{black}{$q_{20}$} & \color{red}{$q_{22}$} & \color{red}{$q_{21}$} &\color{black}{$q_{1}$}& \color{red}{$q_{23}$} & $q_{2}$ & $q_{27}$ & \color{red}{$q_{24}$} & $q_{4}$ & $q_{3}$ & \color{red}{$q_{25}$} & $q_{5}$ & \color{red}{$q_{26} $} & \color{red}{$q_{19}$} & $q_{6}$ & $q_{18}$ & $q_{7}$ & $q_{16}$ & $q_{17}$ & $q_{9}$ & $q_{8}$ & $q_{15}$ & $q_{10}$ & $q_{14}$ & $q_{13}$ & $q_{12}$ & $q_{11}$ \\

21 & \color{red}{$q_{22}$} & \color{red}{$q_{23}$} & \color{black}{$q_{21}$} &\color{black}{$q_{1}$} & \color{red}{$q_{24}$} & $q_{2}$ & $q_{27}$ & \color{red}{$q_{25}$} & $q_{4}$ & $q_{3}$ & \color{red}{$q_{26}$} & $q_{5}$ & \color{red}{$q_{20} $} & $q_{19}$ & $q_{6}$ & $q_{18}$ & $q_{7}$ & $q_{16}$ & $q_{17}$ & $q_{9}$ & $q_{8}$ & $q_{15}$ & $q_{10}$ & $q_{14}$ & $q_{13}$ & $q_{12}$ & $q_{11}$ \\

22 & \color{black}{$q_{22}$} & \color{red}{$q_{24}$} & \color{red}{$q_{23}$} &\color{black}{$q_{1}$}& \color{red}{$q_{25}$} & $q_{2}$ & $q_{27}$ & \color{red}{$q_{26}$} & $q_{4}$ & $q_{3}$ & \color{red}{$q_{21}$} & $q_{5}$ & $q_{20} $ & $q_{19}$ & $q_{6}$ & $q_{18}$ & $q_{7}$ & $q_{16}$ & $q_{17}$ & $q_{9}$ & $q_{8}$ & $q_{15}$ & $q_{10}$ & $q_{14}$ & $q_{13}$ & $q_{12}$ & $q_{11}$ \\

23 & \color{red}{$q_{24}$} & \color{red}{$q_{25}$} & \color{black}{$q_{23}$} &\color{black}{$q_{1}$}& \color{red}{$q_{26}$} & $q_{2}$ & \color{red}{$q_{22}$} & \color{red}{$q_{27}$} & $q_{4}$ & $q_{3}$ & $q_{21}$ & $q_{5}$ & $q_{20} $ & $q_{19}$ & $q_{6}$ & $q_{18}$ & $q_{7}$ & $q_{16}$ & $q_{17}$ & $q_{9}$ & $q_{8}$ & $q_{15}$ & $q_{10}$ & $q_{14}$ & $q_{13}$ & $q_{12}$ & $q_{11}$ \\

24 & \color{black}{$q_{24}$} & \color{red}{$q_{26}$} & \color{red}{$q_{25}$} &\color{black}{$q_{1}$} & \color{red}{$q_{27}$} & $q_{2}$ & $q_{22}$ & \color{red}{$q_{23}$} & $q_{4}$ & $q_{3}$ & $q_{21}$ & $q_{5}$ & $q_{20} $ & $q_{19}$ & $q_{6}$ & $q_{18}$ & $q_{7}$ & $q_{16}$ & $q_{17}$ & $q_{9}$ & $q_{8}$ & $q_{15}$ & $q_{10}$ & $q_{14}$ & $q_{13}$ & $q_{12}$ & $q_{11}$ \\

25 & \color{red}{$q_{26}$} & \color{red}{$q_{25}$} & \color{red}{$q_{24}$} &\color{black}{$q_{1}$}& \color{black}{$q_{27}$} & $q_{2}$ & $q_{22}$ & $q_{23}$ & $q_{4}$ & $q_{3}$ & $q_{21}$ & $q_{5}$ & $q_{20} $ & $q_{19}$ & $q_{6}$ & $q_{18}$ & $q_{7}$ & $q_{16}$ & $q_{17}$ & $q_{9}$ & $q_{8}$ & $q_{15}$ & $q_{10}$ & $q_{14}$ & $q_{13}$ & $q_{12}$ & $q_{11}$ \\

26 & \color{black}{$q_{26}$} & \color{red}{$q_{27}$} & \color{black}{$q_{24}$} &\color{black}{$q_{1}$}& \color{red}{$q_{25}$} & $q_{2}$ & $q_{22}$ & $q_{23}$ & $q_{4}$ & $q_{3}$ & $q_{21}$ & $q_{5}$ & $q_{20} $ & $q_{19}$ & $q_{6}$ & $q_{18}$ & $q_{7}$ & $q_{16}$ & $q_{17}$ & $q_{9}$ & $q_{8}$ & $q_{15}$ & $q_{10}$ & $q_{14}$ & $q_{13}$ & $q_{12}$ & $q_{11}$ \\

\hline
\end{tabular}
\label{tab:qft27}
\end{table}

Analysing data from Table.~\ref{tab:qft27}, we obtain that the CNOT-cost of the 27-qubit implementation on Falcon 5.11 architecture is equal to 957.

\section{Conclusion}
\label{sec:con}
In this work, we presented an approach to optimize the implementation of the quantum hashing algorithm when implemented on real quantum devices with specific architectures of qubits such as 16-qubit Falcon r4p and 27-qubit Falcon 5.11 architectures. The method focuses on reducing a CNOT-cost of the short circuit implemention of quantum hashing algorithm. We executed the circuit on simulators of real devices $\mathtt{ibmq}\_\mathtt{guadelupe}$ and $\mathtt{ibmq}\_\mathtt{auckland}$. The results show that the approach gives approximately 3 times advantage in a number of CNOT gates for the proposed circuit after transpillation comparing with the transpiled original circuit. 

We also applied our approach to QFT algorithm implementation on the aforementioned machines and obtained corresponding CNOT-costs.




\bibliographystyle{splncs04}
\bibliography{_ref}

\end{document}